\documentclass[preprint,12pt]{elsarticle}

\usepackage[T1]{fontenc}
\usepackage{amsmath,amsfonts,amssymb,amsthm}
\usepackage{microtype}

\biboptions{numbers,sort&compress}
\numberwithin{equation}{section}

\newtheorem{theorem}{Theorem}[section]
\newtheorem{proposition}[theorem]{Proposition}
\newtheorem{lemma}[theorem]{Lemma}
\newtheorem{corollary}[theorem]{Corollary}

\theoremstyle{definition}
\newtheorem{definition}[theorem]{Definition}

\theoremstyle{remark}
\newtheorem*{remark}{Remark}

\DeclareMathOperator{\Id}{Id}
\DeclareMathOperator{\Op}{Op}
\DeclareMathOperator{\diver}{div}
\newcommand{\R}{\mathbb R}
\newcommand{\N}{\mathbb N}

\newcommand{\cX}{\mathcal X}
\newcommand{\cP}{\mathcal P}

\newcommand{\cE}{\mathcal E}
\newcommand{\bfx}{\mathbf x}
\newcommand{\bfp}{\mathbf p}
\newcommand{\bF}{\mathbf F}
\newcommand{\bH}{\mathbf H}
\newcommand{\bdx}{\boldsymbol\partial_{\mathbf x}}
\newcommand{\bdp}{\boldsymbol\partial_{\mathbf p}}
\newcommand{\qdx}{\partial_{\mathbf x}^{\,q}}
\newcommand{\qdp}{\partial_{\mathbf p}^{\,q}}
\newcommand{\oneX}{\mathbf 1_{\mathcal X}}
\newcommand{\Ux}{\mathcal U_x}
\newcommand{\Up}{\mathcal U_p}
\newcommand{\Ax}{\mathcal A_x}

\newcommand{\Mx}{\mathcal M_x}
\newcommand{\Jx}{\mathcal J_x}
\newcommand{\Jp}{\mathcal J_p}
\newcommand{\Dp}{\mathcal D_p}
\newcommand{\starh}{\star_h}
\newcommand{\Xstar}{\mathfrak X_{h,\star}^{H}}
\newcommand{\Xjack}{\mathfrak X_{h,\mathrm{cov}}^{H}}
\newcommand{\Bq}{\mathfrak B_h^{H}}
\newcommand{\Vfield}{V_h[H]}
\newcommand{\Vzero}{V_0[H]}

\journal{Journal of Geometry and Physics}

\begin{document}
\begin{frontmatter}

\title{Representation--Symbol Correspondence for $q$-Hamiltonian Mechanics on the Quantum Plane}

\author[swjtu]{Xiaomei Yang}
\ead{yangxiaomath@swjtu.edu.cn}
\author[uestc]{Zhiliang Deng\corref{cor1}}
\ead{dengzhl@uestc.edu.cn}
\cortext[cor1]{Corresponding author.}

\address[swjtu]{School of Mathematics, Southwest Jiaotong University, Chengdu, China}
\address[uestc]{School of Mathematical Science, University of Electronic Science and Technology of China, Chengdu, China}

\begin{abstract}
We develop a representation--symbol correspondence for $q$-Hamiltonian mechanics on the quantum plane. The coordinate algebra and its covariant differential calculus are realized simultaneously on a smooth commutative function space by multiplication, dilation, and Jackson operators. Normal ordering identifies the coordinate algebra with a polynomial symbol space and transports operator composition to an explicit associative star product. Under this correspondence, the induced covariant $q$-derivatives intertwine exactly with the Jackson operators, and the ordered Hamiltonian action descends to an exact star-Jackson action. For the coordinate observables, the relevant Jackson derivatives are constants, so the star factors reduce to the unit and the Jackson coordinate equations are exact symbol images of the formal quantum-plane equations. For nonlinear observables, pointwise multiplication produces a genuine commutativization error. We quantify this distinction through a first-order expansion of the operator quantization map, study the derivation, divergence, and energy defects of the induced pointwise Jackson dynamics, and prove convergence of the represented Hamiltonian, the symbol actions, the Jackson vector field, and its finite-time trajectories to their classical counterparts as $q\to1$. The resulting framework connects covariant quantum-plane differential calculus, deformation products, and Hamiltonian dynamics while keeping exact algebraic statements separate from commutative and classical approximations.
\end{abstract}


\end{frontmatter}

\section{Introduction}
\label{sec:introduction}

Covariant differential calculi provide an algebraic framework for formulating
dynamics on quantum spaces.  Passing from such formal dynamics to concrete
evolution equations, however, requires one to represent the noncommuting
coordinates and their differential calculus by operators acting on ordinary
function spaces.  The quantum plane provides a prototypical setting for this
problem.  Its coordinate generators satisfy
$$
\mathbf p\mathbf x=q\mathbf x\mathbf p,
$$
and the Wess--Zumino calculus equips this algebra with a compatible covariant
$q$-differential structure \cite{WessZumino1991}.  Building on this setting,
$q$-deformed Lagrangian and Hamiltonian formalisms were developed in
\cite{CabanEtAl1994,LukinSternYakushin1993,Lavagno2006}; see also
\cite{KlimykSchmudgen1997,Majid1995,Manin2018} for broader background on
quantum groups and noncommutative geometry.  Since the underlying deformation
acts through multiplicative shifts, Jackson differences
\cite{KacCheung2002} provide a natural calculus for concrete realizations of
the corresponding differential operators.

Two complementary lines of work are particularly relevant to this passage
from algebraic relations to concrete dynamics.  On the operator side,
differential operators on noncommutative rings and on the quantum plane were
studied in \cite{LuntsRosenberg1997,IyerMcCune2003}, while explicit connections
between covariant quantum-plane differential calculi and $q$-difference or
scaling operators were developed in
\cite{Ubriaco1992,BauerWachter2003}.  Operator representations of
quantum-plane relations were investigated in \cite{Schmudgen2002}, resolvent
formulations for the real quantum plane in \cite{OstrovskyiSchmudgen2014},
and representation-theoretic classical limits in \cite{IpClassical2013}.
On the symbol side, universal deformation formulas and twist constructions
provide general mechanisms for deforming multiplication
\cite{GiaquintoZhang1998}, and covariant star-product realizations of quantum
spaces, including the quantum plane, were obtained in \cite{Blohmann2003}.
Related deformation-theoretic and Hochschild-cohomological properties of the
quantum plane and the $q$-Weyl algebra were studied in
\cite{GerstenhaberGiaquinto2014}, while shifted Leibniz rules arising in
$q$-difference calculi fit into the general theory of $\sigma$-derivations
\cite{HartwigLarssonSilvestrov2006}.  Thus one has, on the one hand, concrete
operator realizations of quantum-plane relations and, on the other, algebraic
mechanisms for transporting noncommutative multiplication to a commutative
symbol space.

The interaction between these two levels becomes especially relevant when a
Hamiltonian action is introduced.  Related realization questions occur in
other parts of $q$-deformed mathematical physics; for example, Melong and
Wulkenhaar \cite{MelongWulkenhaar2026} study $q$-deformed Airy structures
through $q$-difference operators, together with $q$-quantized spectral curves
and shifted $q$-difference equations.  A more direct motivation for the
present work comes from the $q$-analogue of Hamiltonian Monte Carlo in
\cite{YangDeng2025}, where formal quantum-plane Hamiltonian equations are
converted into a pointwise system involving dilations and Jackson
differences.  The sampling algorithm itself is not studied here.  Rather, it
raises a structural question common to the preceding developments: can the
coordinate algebra, the covariant differential calculus, the operator
realization, and the induced symbol product be incorporated into a single
framework in which the algebraic correspondence remains exact before any
pointwise commutative approximation is introduced?

We construct such a framework through a representation--symbol
correspondence.  On the positive quadrant, coordinatewise dilations form
commuting one-parameter groups generated by Euler operators.  Combining
multiplication with dilation gives a faithful representation of the
quantum-plane coordinate algebra, while the same dilation groups generate
Jackson operators realizing the covariant derivatives.  Normal ordering then
identifies the coordinate algebra, as a vector space, with
$$
\mathcal P=\mathbb C[x,p].
$$
Transporting operator composition through this identification yields the
associative normal star product
$$
F\star_h G
=
\sum_{k=0}^{\infty}
\frac{h^k}{k!}
\left(\mathcal A_p^kF\right)
\left(\mathcal A_x^kG\right),
$$
which is exact for polynomial symbols and satisfies
$$
p\star_h x=qxp.
$$
This product is not introduced as an independent deformation quantization;
rather, it records on the commutative symbol space the multiplication already
present in the represented quantum-plane algebra.

The main algebraic consequence is the compatibility of this symbol calculus
with the ordered Hamiltonian action.  The induced $q$-derivatives intertwine
exactly with the corresponding Jackson operators, while ordered products
descend to the normal star product.  Hence the formal Hamiltonian action has
an exact star--Jackson symbol image.  For the coordinate observables $x$ and
$p$, the relevant Jackson derivatives are constants, so the star factors
reduce to the unit and the Jackson coordinate equations are obtained exactly,
without a small-$h$ approximation.  The distinction between exact and
commutative levels becomes visible only for nonlinear observables, where
replacing the star product by pointwise multiplication changes the action.
We compare the exact star action with covariant and asymmetric pointwise
Jackson actions, derive the associated shifted product rules and derivation
defects, and study the ordinary divergence and energy variation of the
resulting pointwise vector field.  These are properties of the commutative
realization and are kept distinct from covariance properties of the
noncommutative differential calculus.  Finally, as
$$
q=e^h\longrightarrow1,
$$
we establish local convergence of the represented Hamiltonian, the symbol
actions, and the Jackson vector field to their classical counterparts, as
well as uniform finite-time convergence of the corresponding trajectories
under the stated compactness assumptions.

The exact algebraic statements are formulated for polynomial Hamiltonians and
observables, for which normal ordering and operator composition are finite and
unambiguous, while the asymptotic results are established locally for
sufficiently smooth functions.  No braided symplectic interpretation of the
pointwise vector field or classification of quantum-plane representations is
assumed.  The paper is organized as follows.
Section~\ref{sec:dilation} constructs the dilation representation of the
coordinate algebra and its covariant differential calculus.
Section~\ref{sec:symbols} develops the normal-symbol calculus and establishes
the exact representation--symbol correspondence for the Hamiltonian action.
Section~\ref{sec:euclidean} compares the exact star-symbol dynamics with
pointwise Jackson realizations and studies their derivation and geometric
properties.  Section~\ref{sec:limit} establishes the classical and finite-time
limits.

\section{Dilation representation of the quantum-plane differential calculus}
\label{sec:dilation}

The first step is to place the coordinate algebra and its covariant
differential calculus in a single operator framework.  Coordinatewise
dilations provide the underlying group action, multiplication--dilation
operators realize the quantum-plane relation faithfully, and the same
dilation groups generate the Jackson operators representing the covariant
$q$-derivatives.  The point of the construction is therefore not to postulate
coordinate and derivative substitutions separately, but to obtain them as
the images of generators of one represented differential algebra.

\subsection{Dilation groups and logarithmic coordinates}
\label{subsec:dilation-groups}

Let $\mathbb R_+^2=(0,\infty)\times(0,\infty)$,
and set $\mathcal X=C^\infty(\mathbb R_+^2;\mathbb C)$. 
We denote by $\oneX\in\mathcal X$
the constant function $\oneX(x, p)=1$.
This symbol is reserved for the constant function in $\mathcal X$ and is
therefore distinct from the algebraic units of $\mathfrak A_q$ and
$\mathfrak D_q$ introduced below.

We equip $\mathcal X$ with its standard Fr\'echet topology. More precisely,
for every compact set $K\Subset\mathbb R_+^2$ and every nonnegative integer
$r$, define the seminorm
\begin{align*}
\|F\|_{K, r}&:=\max_{|\gamma|\le r}\sup_{(x, p)\in K}
\left|\partial^\gamma F(x, p)\right|.
\end{align*}
The topology generated by these seminorms is the topology of uniform
convergence on compact subsets of all partial derivatives.

For $t\in\mathbb R$, define the coordinatewise dilation operators
$\mathcal U_x(t)$ and $\mathcal U_p(t)$ by
\begin{align*}
\bigl(\mathcal U_x(t)F \bigr)(x, p)=F(e^t x, p),\qquad
\bigl(\mathcal U_p(t)F\bigr)(x, p)=F(x, e^t p).
\end{align*}

\begin{proposition}
\label{prop:dilation-groups}
The families $\left\{\mathcal U_x(t)\right\}_{t\in\mathbb R}$ and $\left\{\mathcal U_p(t)\right\}_{t\in\mathbb R}$
are strongly continuous one-parameter groups of continuous linear operators
on $\mathcal X$. They commute:
\begin{align*}
\mathcal U_x(t)\mathcal U_p(s)
&=\mathcal U_p(s)\mathcal U_x(t)
\end{align*}
for all $s,t\in\mathbb R$. Their infinitesimal generators are defined on all
of $\mathcal X$ and are given by the Euler operators
\begin{align*}
\mathcal A_x=x\partial_x,\qquad
\mathcal A_p=p\partial_p.
\end{align*}
\end{proposition}

\begin{proof}
The group laws, inverses, and commutation follow directly from the definitions of
$\mathcal U_x$ and $\mathcal U_p$.  For $t$ in a compact interval,
the image of a fixed compact set $K\Subset\mathbb R_+^2$ under either
dilation is contained in a common compact set.  The chain rule therefore
gives continuity in every defining seminorm, and the smooth
dependence on $t$ gives strong continuity.  Finally,
\begin{align*}
\lim_{t\to0}\frac{\mathcal U_x(t)F-F}{t}
&=x\partial_xF,
&\lim_{t\to0}\frac{\mathcal U_p(t)F-F}{t}
&=p\partial_pF
\end{align*}
in $\mathcal X$, which identifies the generators.
\end{proof}

In terms of the infinitesimal generators, we use the notation
\begin{align*}
\mathcal U_x(t)=e^{t\mathcal A_x},\qquad
\mathcal U_p(t)=e^{t\mathcal A_p}.
\end{align*}
Here the exponential notation refers to the strongly continuous
one-parameter groups generated by $\mathcal A_x$ and $\mathcal A_p$.

We next describe the dilation groups in logarithmic coordinates. Let
$\mathcal Y=C^\infty(\mathbb R^2;\mathbb C)$, equipped with its standard
Fr\'echet topology, and define
\begin{align*}
\bigl(\mathcal CF\bigr)(u,v)=F(e^u,e^v)=:G(u,v),\qquad
\bigl(\mathcal C^{-1}G\bigr)(x,p)=G(\log x,\log p).
\end{align*}

\begin{proposition}
\label{prop:logarithmic-conjugacy}
The map $\mathcal C:\mathcal X\longrightarrow\mathcal Y$
is a linear topological isomorphism. Under this isomorphism, the dilation
groups are conjugate to the translation groups:
\begin{align*}
\bigl(
\mathcal C\mathcal U_x(t)\mathcal C^{-1}G
\bigr)(u, v)=G(u+t, v),\qquad
\bigl(\mathcal C\mathcal U_p(t)\mathcal C^{-1}G\bigr)(u, v)=G(u, v+t).
\end{align*}
Consequently,
\begin{align*}
\mathcal C\mathcal A_x\mathcal C^{-1}=\partial_u,\qquad
\mathcal C\mathcal A_p\mathcal C^{-1}=\partial_v.
\end{align*}
\end{proposition}

\begin{proof}
The exponential map $(u, v)\mapsto(e^u, e^v)$ is a smooth diffeomorphism from
$\mathbb R^2$ onto $\mathbb R_+^2$, so pullback gives the stated topological
isomorphism.  Direct substitution yields
\begin{align*}
\bigl(\mathcal C\mathcal U_x(t)\mathcal C^{-1}G\bigr)(u, v)
&=G(u+t, v),
&\bigl(\mathcal C\mathcal U_p(t)\mathcal C^{-1}G\bigr)(u, v)
&=G(u, v+t).
\end{align*}
Differentiation at $t=0$ gives the asserted generator identities.
\end{proof}

Thus multiplicative displacement in the original coordinates becomes
ordinary additive translation in logarithmic coordinates. In particular, the
Euler operators are the infinitesimal generators of scale changes in the
original variables and ordinary translation generators after the logarithmic
change of coordinates.

\subsection{Coordinate algebra and faithful dilation representation}
\label{subsec:coordinate-representation}

We now turn to the coordinate algebra of the quantum plane. Fix $q=e^h>0$. 
Let $\mathbb C\langle\mathbf x, \mathbf p\rangle$ denote the free unital associative algebra over $\mathbb C$ generated by two
noncommuting symbols $\mathbf x$ and $\mathbf p$. Define
\begin{align*}
\mathfrak A_q
&=\mathbb C\langle\mathbf x, \mathbf p\rangle
\Big/\left\langle\mathbf p\mathbf x-q\mathbf x\mathbf p\right\rangle,
\end{align*}
where $\left\langle\mathbf p\mathbf x-q\mathbf x\mathbf p\right\rangle$
denotes the two-sided ideal generated by
$\mathbf p\mathbf x-q\mathbf x\mathbf p$.
We continue to denote the equivalence classes of the two free generators by
$\mathbf x$ and $\mathbf p$. Thus, in $\mathfrak A_q$,
\begin{align*}
\mathbf p\mathbf x
&=
q\mathbf x\mathbf p.
\end{align*}
The boldface symbols $\mathbf x$ and $\mathbf p$ are abstract
noncommuting algebra generators. They must be distinguished from the ordinary
commuting coordinate functions $x$ and $p$ on $\mathbb R_+^2$.

Let $\mathcal L(\mathcal X)$ denote the unital algebra of continuous linear
maps from $\mathcal X$ into itself, with multiplication given by operator
composition. For every $a\in\mathcal X$, let
$$
\mathcal M_a:\mathcal X\longrightarrow\mathcal X
$$
denote multiplication by $a$:
\begin{align*}
\bigl(
\mathcal M_aF
\bigr)(x,p)
&=
a(x,p)F(x,p).
\end{align*}
The bilinear multiplication map on $\mathcal X$ is continuous in its standard
Fr\'echet topology. Hence $\mathcal M_a\in\mathcal L(\mathcal X)$ for every
$a\in\mathcal X$, and
\begin{align*}
\mathcal M_a\mathcal M_b=\mathcal M_{ab},\qquad
\mathcal M_{\oneX}=\Id_{\mathcal X}.
\end{align*}
In particular, the coordinate multiplication operators are
\begin{align*}
\bigl(
\mathcal M_xF
\bigr)(x, p)=xF(x, p),\qquad
\bigl(\mathcal M_pF\bigr)(x, p)=pF(x, p).
\end{align*}
Thus
\begin{align*}
\mathcal M_x^m\mathcal M_p^n
&=
\mathcal M_{x^mp^n}
\end{align*}
for all $m,n\in\mathbb N_0$. In particular,
$\mathcal M_x\mathcal M_p=\mathcal M_p\mathcal M_x$.

Define
\begin{align*}
\Pi_{\mathbf x}:=\mathcal M_x,\qquad
\Pi_{\mathbf p}:=\mathcal M_p\mathcal U_x(h).
\end{align*}
Their pointwise actions are
\begin{align*}
\bigl(
\Pi_{\mathbf x}F
\bigr)(x, p)=xF(x, p),\qquad
\bigl(\Pi_{\mathbf p}F
\bigr)(x, p)=pF(qx, p).
\end{align*}

\begin{proposition}
\label{prop:coordinate-representation}
The operators $\Pi_{\mathbf x}$ and $\Pi_{\mathbf p}$ satisfy $\Pi_{\mathbf p}\Pi_{\mathbf x}=q\Pi_{\mathbf x}\Pi_{\mathbf p}$.
Consequently, the assignment
$$
\mathbf x
\longmapsto
\Pi_{\mathbf x},
\qquad
\mathbf p
\longmapsto
\Pi_{\mathbf p}
$$
extends uniquely to a unital algebra representation
\begin{align*}
\Pi:
\mathfrak A_q
\longrightarrow
\mathcal L(\mathcal X).
\end{align*}
\end{proposition}

\begin{proof}
Directly from the definitions,
\begin{align*}
\mathcal U_x(h)\mathcal M_x
&=q\mathcal M_x\mathcal U_x(h).
\end{align*}
Since $\mathcal M_p$ commutes with $\mathcal M_x$ and $\mathcal U_x(h)$,
\begin{align*}
\Pi_{\mathbf p}\Pi_{\mathbf x}
&=\mathcal M_p\mathcal U_x(h)\mathcal M_x
=q\mathcal M_x\mathcal M_p\mathcal U_x(h)
=q\Pi_{\mathbf x}\Pi_{\mathbf p}.
\end{align*}
Hence the defining ideal of $\mathfrak A_q$ lies in the kernel of the
homomorphism from the free algebra determined by the two operators, and the
quotient universal property gives the representation $\Pi$.
\end{proof}

The covariance relation above may equivalently be written as
\begin{align*}
\mathcal U_x(h)
\mathcal M_x
\mathcal U_x(-h)
&=
q\mathcal M_x.
\end{align*}
Thus $\mathcal M_x$ has homogeneity one under the $x$-dilation. The
noncommutativity of $\Pi_{\mathbf x}$ and $\Pi_{\mathbf p}$ is generated by
the interaction between multiplication and dilation; it does not arise from
any noncommutativity of the ordinary variables $x$ and $p$.

We next record the normal-ordering property of $\mathfrak A_q$ and the
faithfulness of the preceding representation.

\begin{theorem}[Normal ordering and faithfulness]
\label{thm:normal-ordering-faithfulness}
The ordered monomials
$\left\{\mathbf x^m\mathbf p^n: m, n\in\mathbb N_0\right\}$
form a vector-space basis of $\mathfrak A_q$, with
$\mathbf x^0\mathbf p^0=1_{\mathfrak A_q}$. Consequently, every
$\mathbf F\in\mathfrak A_q$ has a unique finite expansion
\begin{align*}
\mathbf F
&=
\sum_{m,n\ge0}
a_{mn}\mathbf x^m\mathbf p^n.
\end{align*}
Moreover, the representation
$\Pi: \mathfrak A_q\longrightarrow\mathcal L(\mathcal X)$ is injective.
\end{theorem}

\begin{proof}
Repeated use of $\mathbf p\mathbf x=q\mathbf x\mathbf p$ decreases the number
of inversions in a word and therefore rewrites every word as a scalar multiple
of some $\mathbf x^m\mathbf p^n$.  Thus the ordered monomials span
$\mathfrak A_q$.  Moreover,
$\Pi\left(\mathbf x^m\mathbf p^n\right)=\mathcal M_{x^mp^n}\mathcal U_x(nh)$,
and hence
$\Pi\left(\mathbf x^m\mathbf p^n\right)\oneX=x^mp^n$.
A linear relation among the ordered monomials would therefore give the same
linear relation among the ordinary monomials $x^mp^n$ on the open set
$\mathbb R_+^2$, which is impossible unless all coefficients vanish.  This
proves the basis statement.  Finally, if $\mathbf F\in\ker\Pi$, expand
$\mathbf F$ in that basis and apply $\Pi(\mathbf F)$ to $\oneX$; linear
independence of the ordinary monomials again forces every coefficient to
vanish.  Thus $\Pi$ is injective.
\end{proof}

\begin{corollary}
\label{cor:concrete-coordinate-representation}
Let
$$
\mathbf F=\sum_{m, n\ge0}a_{mn}\mathbf x^m\mathbf p^n
$$
be the unique normal expansion of $\mathbf F\in\mathfrak A_q$. Then
\begin{align*}
\Pi(\mathbf F)
&=\sum_{m, n\ge0}a_{mn}\mathcal M_{x^m p^n}\mathcal U_x(n h).
\end{align*}
Equivalently, for every $G\in\mathcal X$,
\begin{align*}
\bigl(\Pi(\mathbf F)G\bigr)(x, p)
&=\sum_{m, n\ge0}a_{mn}x^mp^nG(q^n x, p).
\end{align*}
In particular,
\begin{align*}
\operatorname{Im}\Pi
&=
\operatorname{alg}
\left(
\Id_{\mathcal X},
\Pi_{\mathbf x},
\Pi_{\mathbf p}
\right),
\end{align*}
the unital operator subalgebra of $\mathcal L(\mathcal X)$ generated by
$\Pi_{\mathbf x}$ and $\Pi_{\mathbf p}$.
\end{corollary}

\begin{proof}
The operator formula follows from linearity, multiplicativity, and the
ordered-monomial representation above. Its pointwise form is immediate from
$$
\bigl(
\mathcal U_x(nh)G
\bigr)(x,p)
=
G(q^nx,p).
$$
The description of the image follows because $\mathfrak A_q$ is generated as
a unital algebra by $\mathbf x$ and $\mathbf p$.
\end{proof}

The identity
$$
\Pi(\mathbf F)\oneX
=
\sum_{m,n}a_{mn}x^mp^n
$$
is an equality in $\mathcal X$. It will be used in
Section~\ref{sec:symbols} to recover the normal symbol of a represented
coordinate polynomial.

\begin{remark}
\label{rem:sign-sectors}
The choice of $\mathbb R_+^2$ is convenient because it permits global
logarithmic coordinates and makes the inverse multiplication operators used
in the Jackson calculus nonsingular. Operator constructions on the quantum
quarter plane and the real quantum plane provide a broader analytic context
for such positive and signed realizations; see \cite{Schmudgen2002}. The
algebraic representation itself is not restricted to the positive quadrant.

For
$\varepsilon_x,\varepsilon_p\in\{-1,1\}$, define the open sign sector
$$\Omega_{\varepsilon_x,\varepsilon_p}=
\left\{(x, p)\in\mathbb R^2: \varepsilon_xx>0,\, \varepsilon_p p>0\right\}.$$
Positive dilations preserve each such sector. Replacing $\mathcal X$ by
$C^\infty
\left(
\Omega_{\varepsilon_x, \varepsilon_p};\mathbb C
\right)$
and retaining the same multiplication and dilation formulas gives the same
operator identities and the same faithful representation of
$\mathfrak A_q$. Logarithmic coordinates on that sector are given by
$$
u=\log|x|,
\qquad
v=\log|p|.
$$
The coordinate axes are excluded because they are not preserved as open
multiplicative coordinate charts and because inverse multiplication operators
used later in the Jackson calculus become singular there. Values on the axes
may be introduced separately by continuous extension when such extensions
exist.
\end{remark}

\subsection{Covariant differential calculus and Jackson realization}
\label{subsec:differential}

The faithful coordinate representation does not yet determine how the
covariant derivative generators act.  We now extend the same operator
framework to the differential calculus.  The relevant Jackson operators are
generated by the dilation groups, the mixed coordinate--derivative relations
are verified at the operator level, and the induced left $q$-derivatives on
the coordinate algebra are then recovered from this representation.

Assume throughout that
$$
q=e^h>0,
\qquad
q\ne1.
$$
Define the forward Jackson operators on $\mathcal X$ by
\begin{align*}
\mathcal J_x
&:=
\frac{1}{q^2-1}
\mathcal M_x^{-1}
\bigl(
\mathcal U_x(2h)-\Id_{\mathcal X}
\bigr),\\
\mathcal J_p
&:=
\frac{1}{q^2-1}
\mathcal M_p^{-1}
\bigl(
\mathcal U_p(2h)-\Id_{\mathcal X}
\bigr).
\end{align*}
Since $x, p>0$ on $\mathbb R_+^2$,
$\mathcal M_x^{-1}=\mathcal M_{1/x}$,
$\mathcal M_p^{-1}=\mathcal M_{1/p}$,
and both inverse multiplication operators belong to
$\mathcal L(\mathcal X)$. Pointwise,
\begin{align*}
(\mathcal J_xF)(x,p)
&=
\frac{F(q^2x,p)-F(x,p)}{(q^2-1)x},\\
(\mathcal J_pF)(x,p)
&=
\frac{F(x,q^2p)-F(x,p)}{(q^2-1)p}.
\end{align*}
The coupled $p$-direction operator and its pointwise action are
\begin{align*}
\mathcal D_p:=\mathcal J_p\mathcal U_x(h),\qquad
(\mathcal D_pF)(x, p)=\frac{F(qx, q^2p)-F(qx, p)}{(q^2-1)p}.
\end{align*}
The preliminary $x$-dilation in $\mathcal D_p$ is the differential
counterpart of the dilation contained in
$$
\Pi_{\mathbf p}=\mathcal M_p\mathcal U_x(h).
$$

We now introduce the specific left covariant calculus underlying the formal
Hamiltonian model.  It is the quantum-plane calculus developed from the
Wess--Zumino framework and its $q$-analysis realizations
\cite{WessZumino1991,Ubriaco1992,Lavagno2006}.  To state its compatibility
with the Jackson realization compactly, let $\mathfrak D_q$ denote the unital
associative algebra generated by $\mathbf x$, $\mathbf p$, $\boldsymbol\partial_{\mathbf x}$, $\boldsymbol\partial_{\mathbf p}$,
subject to
\begin{align}
\begin{aligned}
&\mathbf p\mathbf x=q\mathbf x\mathbf p,\,\,
\boldsymbol\partial_{\mathbf p}\mathbf x=q\mathbf x\boldsymbol\partial_{\mathbf p},\,\,
\boldsymbol\partial_{\mathbf x}\mathbf p=q\mathbf p\boldsymbol\partial_{\mathbf x},\,\,
\boldsymbol\partial_{\mathbf x}\mathbf x=1_{\mathfrak D_q}+q^2\mathbf x\boldsymbol\partial_{\mathbf x},\\
&\boldsymbol\partial_{\mathbf p}\mathbf p=1_{\mathfrak D_q}+q^2\mathbf p\boldsymbol\partial_{\mathbf p}+(q^2-1)\mathbf x\boldsymbol\partial_{\mathbf x}, \,\,
\boldsymbol\partial_{\mathbf p}\boldsymbol\partial_{\mathbf x}=q^{-1}\boldsymbol\partial_{\mathbf x}
\boldsymbol\partial_{\mathbf p}.
\end{aligned}
\label{eq:Dq-relations}
\end{align}
We use these left covariant relations in the form adopted in
\cite{YangDeng2025}.  Here $\mathfrak D_q$ serves only to package the
coordinate and derivative generators into one compatibility algebra; after
the representation theorem below, the analysis proceeds through the induced
$q$-derivatives on $\mathfrak A_q$ and their Jackson realizations.

\begin{theorem}[Representation of the covariant differential algebra]
\label{thm:full-differential-representation}
The four operators $\mathcal M_x$,
$\mathcal M_p\mathcal U_x(h)$,
$\mathcal J_x$,
$\mathcal D_p$
satisfy the defining relations \eqref{eq:Dq-relations}. Consequently, the
assignment
\begin{align*}
\begin{aligned}
&\mathbf x\longmapsto\mathcal M_x, &&\mathbf p\longmapsto\mathcal M_p\mathcal U_x(h),\\
&\boldsymbol\partial_{\mathbf x}\longmapsto\mathcal J_x, && \boldsymbol\partial_{\mathbf p}\longmapsto\mathcal D_p
\end{aligned}
\end{align*}
extends uniquely to a unital algebra representation of $\mathfrak D_q$ on
$\mathcal X$.  Moreover, the natural homomorphism
$\mathfrak A_q\longrightarrow\mathfrak D_q$ induced by the two coordinate
generators is injective.  After identifying $\mathfrak A_q$ with its image in
$\mathfrak D_q$, the restriction of the differential-algebra representation is
the faithful coordinate representation $\Pi$ constructed above.
\end{theorem}

\begin{proof}
The coordinate relation is Proposition~\ref{prop:coordinate-representation}.
We verify the remaining operator identities.
Since $\mathcal J_p$ commutes with $\mathcal M_x$ and
$\mathcal U_x(h)\mathcal M_x=q\mathcal M_x\mathcal U_x(h)$,
\begin{align*}
\mathcal D_p\mathcal M_x=
q\mathcal M_x\mathcal D_p.
\end{align*}
Moreover, direct use of the definitions gives
\begin{align*}
\mathcal J_x\mathcal U_x(h)=
q\mathcal U_x(h)\mathcal J_x.
\end{align*}
Because $\mathcal M_p$ commutes with both factors,
\begin{align*}
\mathcal J_x\Pi_{\mathbf p}=
q\Pi_{\mathbf p}\mathcal J_x.
\end{align*}
The one-variable Jackson identities are
\begin{align*}
\mathcal J_x\mathcal M_x=
\Id_{\mathcal X}+q^2\mathcal M_x\mathcal J_x,\qquad
\mathcal J_p\mathcal M_p=\Id_{\mathcal X}+q^2\mathcal M_p\mathcal J_p.
\end{align*}
Using the commutation of $\mathcal U_x(h)$ with $\mathcal M_p$ and
$\mathcal J_p$, together with
$$
\mathcal U_x(2h)=
\Id_{\mathcal X}+(q^2-1)\mathcal M_x\mathcal J_x,
$$
we obtain
\begin{align*}
\mathcal D_p\Pi_{\mathbf p}=
\Id_{\mathcal X}+q^2\Pi_{\mathbf p}\mathcal D_p+(q^2-1)\mathcal M_x\mathcal J_x.
\end{align*}
Finally, $\mathcal J_x$ commutes with $\mathcal J_p$, and the preceding
covariance identity is equivalent to
$$
\mathcal U_x(h)\mathcal J_x=q^{-1}\mathcal J_x\mathcal U_x(h).
$$
Therefore,
\begin{align*}
\mathcal D_p\mathcal J_x=q^{-1}\mathcal J_x\mathcal D_p.
\end{align*}
These identities are exactly the operator images of the differential relations in
\eqref{eq:Dq-relations}. The universal property of the
quotient algebra gives the asserted representation.  On the coordinate
generators it agrees with $\Pi$.  If an element of $\mathfrak A_q$ mapped to
zero under the natural homomorphism into $\mathfrak D_q$, its image under this
operator representation would therefore be zero under $\Pi$; the faithfulness
of $\Pi$ then forces the original coordinate element to vanish.  Hence the
natural homomorphism $\mathfrak A_q\to\mathfrak D_q$ is injective.
\end{proof}

Thus multiplication, dilation, and Jackson operators realize the coordinate
and derivative generators simultaneously.  In particular, the concrete
formulas that appear in pointwise $q$-Hamiltonian equations are consequences
of an algebra representation rather than independent substitution rules.

For later use, define
\begin{align*}
[z]_{q^2}
&:=
\frac{q^{2z}-1}{q^2-1},
\qquad z\in\mathbb C.
\end{align*}
The PBW basis from Theorem~\ref{thm:normal-ordering-faithfulness} allows us to
define the induced left $q$-derivatives on $\mathfrak A_q$ without retaining
the ambient algebra $\mathfrak D_q$ in the notation.

\begin{definition}[Induced left $q$-derivatives]
\label{def:induced-q-derivatives}
We reserve $\bdx$ and $\bdp$ for the derivative generators of
$\mathfrak D_q$. Their induced left actions on the coordinate algebra are
denoted by the linear maps
\begin{align*}
\qdx,\qdp:
\mathfrak A_q
\longrightarrow
\mathfrak A_q.
\end{align*}
They are defined on the PBW basis by
\begin{align}
\qdx
\left(
\mathbf x^m\mathbf p^n
\right)
&=
\begin{cases}
[m]_{q^2}\mathbf x^{m-1}\mathbf p^n,
& m\ge1,\\
0,
& m=0,
\end{cases}
\notag\\
\qdp
\left(
\mathbf x^m\mathbf p^n
\right)
&=
\begin{cases}
q^m[n]_{q^2}\mathbf x^m\mathbf p^{n-1},
& n\ge1,\\
0,
& n=0,
\end{cases}
\label{eq:induced-q-derivatives}
\end{align}
for all $m,n\in\mathbb N_0$, and then extended linearly.
\end{definition}

These maps are the derivative-free components of the covariant left action
on the ordered coordinate monomials.  Indeed, if one starts with
$\bdx\mathbf x^m\mathbf p^n$ or
$\bdp\mathbf x^m\mathbf p^n$ and repeatedly uses the defining relations to
move the single derivative generator to the right, the derivative-free terms
are exactly those in \eqref{eq:induced-q-derivatives}.  This
observation is all that is required below; no PBW statement for the full
differential algebra $\mathfrak D_q$ is assumed.  In particular,
$$
\qdx 1_{\mathfrak A_q}=0,
\qquad
\qdp 1_{\mathfrak A_q}=0.
$$

\section{Normal symbols and the representation--symbol correspondence}
\label{sec:symbols}

We now pass from the represented noncommutative algebra to a commutative
symbol space without discarding the noncommutative product.  Normal ordering
produces a vector-space identification with polynomial symbols, while
operator composition becomes an explicit associative star product.  This
symbol calculus permits the covariant derivatives and the ordered
$q$-Hamiltonian action to be transported exactly.  The resulting
representation--symbol theorem is the central algebraic statement of the
paper.

\subsection{Normal symbols, quantization, and the induced star product}
\label{subsec:normal-symbols}

Let
$$
\mathcal P
=
\mathbb C[x,p]
$$
be the commutative polynomial algebra. Every $F\in\mathcal P$ has a unique
finite expansion
\begin{align*}
F(x,p)
&=
\sum_{m,n\ge0}c_{mn}(F)x^mp^n.
\end{align*}
For $F\ne0$, define its degree in the $p$ variable by
\begin{align*}
N_p(F)
&:=
\max
\left\{
n:
c_{mn}(F)\ne0
\text{ for some }m
\right\},
\end{align*}
and set $N_p(0)=0$. This notation will be used throughout the operator and
classical-limit estimates below.

We regard the variables $x$ and $p$
here as the ordinary coordinate functions on $\mathbb R_+^2$. Let
\begin{align*}
\iota:
\mathcal P
\longrightarrow
\mathcal X
\end{align*}
denote the natural injective algebra homomorphism that regards a polynomial as
a smooth function on $\mathbb R_+^2$. We shall usually identify
$\mathcal P$ with the subalgebra $\iota(\mathcal P)\subset\mathcal X$, but
retain $\iota$ whenever the ambient space is relevant.

Define the normal-ordering map
\begin{align*}
\mathcal N_q:
\mathcal P
\longrightarrow
\mathfrak A_q
\end{align*}
on monomials by
\begin{align*}
\mathcal N_q(x^mp^n)
&=
\mathbf x^m\mathbf p^n,
\qquad
m,n\in\mathbb N_0,
\end{align*}
and extend it linearly. By
Theorem~\ref{thm:normal-ordering-faithfulness}, the ordered monomials
$\mathbf x^m\mathbf p^n$ form a vector-space basis of $\mathfrak A_q$.
Consequently, $\mathcal N_q$ is a vector-space isomorphism. Its inverse
\begin{align*}
\sigma_N:
\mathfrak A_q
\longrightarrow
\mathcal P
\end{align*}
is called the normal-symbol map.
The corresponding operator quantization map is
\begin{align*}
\Op_h
&:=
\Pi\circ\mathcal N_q:
\mathcal P
\longrightarrow
\mathcal L(\mathcal X).
\end{align*}
Thus the passage
$$
F
\longmapsto
\mathcal N_q(F)
\longmapsto
\Op_h(F)
$$
first replaces a commutative polynomial by its normally ordered
quantum-plane element and then applies the faithful dilation representation.

%Recall that $\oneX\in\mathcal X$
%denotes the constant function $\oneX(x, p)=1$.

\begin{proposition}
\label{prop:explicit-quantization}
For every $m, n\in\mathbb N_0$,
\begin{align}
\Op_h(x^mp^n)
&=
\mathcal M_{x^mp^n}\mathcal U_x(nh).
\label{eq:Op-monomial}
\end{align}
Consequently, if $F(x, p)=\sum_{m,n}c_{mn}x^mp^n$,
then, for every $\psi\in\mathcal X$,
\begin{align}
\bigl(\Op_h(F)\psi\bigr)(x, p)=\sum_{m, n}
c_{mn}x^mp^n\psi(q^nx, p).
\label{eq:Op-action}
\end{align}
Moreover,
$\Op_h(F)\oneX=\iota(F)$.
In particular, $\Op_h$ is injective.
\end{proposition}

\begin{proof}
By Corollary~\ref{cor:concrete-coordinate-representation},
\begin{align*}
\Op_h(x^mp^n)=
\Pi\left(
\mathbf x^m\mathbf p^n
\right)=
\mathcal M_{x^mp^n}\mathcal U_x(nh),
\end{align*}
which proves \eqref{eq:Op-monomial}. Formula \eqref{eq:Op-action} follows by
linearity.
Since every dilation fixes the constant function,
$$
\mathcal U_x(nh)\oneX=\oneX.
$$
Therefore,
\begin{align*}
\Op_h(F)\oneX=\sum_{m, n}c_{mn}\mathcal M_{x^m p^n}\mathcal U_x(nh)\oneX=
\sum_{m,n}c_{mn}x^mp^n=\iota(F),
\end{align*}
where the last equality is an equality in $\mathcal X$.

If $\Op_h(F)=0$, then
$\Op_h(F)\oneX=0$, and hence
$\iota(F)=0$. Since $\iota$ is injective, $F=0$. Thus
$\Op_h$ is injective.
\end{proof}

Since $\Op_h$ is injective, it is a vector-space isomorphism from
$\mathcal P$ onto $\operatorname{Im}(\Op_h)$. We denote its inverse by
\begin{align*}
\sigma_0:=
\Op_h^{-1}: \operatorname{Im}(\Op_h)
\longrightarrow
\mathcal P.
\end{align*}
This inverse has the explicit realization
\begin{align*}
\sigma_0(T)
&=
\iota^{-1}
\left(
T\oneX
\right).
\end{align*}
Indeed, if $T=\Op_h(F)$, then
$$
T\oneX
=
\Op_h(F)\oneX
=
\iota(F),
$$
because every dilation factor fixes the constant function. Consequently,
\begin{align*}
\sigma_0\circ\Op_h=\Id_{\mathcal P},\qquad
\Op_h\circ\sigma_0=\Id_{\operatorname{Im}(\Op_h)}.
\end{align*}
Moreover, since
$\Pi=\Op_h\circ\sigma_N$ on $\mathfrak A_q$,
we obtain
\begin{align*}
\sigma_0\circ\Pi=\sigma_N
\,\,
\text{on }\mathfrak A_q.
\end{align*}
Thus the normal symbol can be recovered either algebraically through
$\sigma_N$ or from the represented operator through $\sigma_0$.

\begin{definition}[Normal star product]
\label{def:star-product}
For $F, G\in\mathcal P$, define
\begin{align*}
F\star_hG
&:=
\sigma_N
\left(
\mathcal N_q(F)\mathcal N_q(G)
\right).
\end{align*}
Equivalently,
\begin{align}
\Op_h(F\star_hG)
&=
\Op_h(F)
\Op_h(G).
\label{eq:star-operator}
\end{align}
\end{definition}

The equivalence follows from
$\Op_h=\Pi\circ\mathcal N_q$, the multiplicativity of $\Pi$,
and the injectivity of $\Op_h$. The product $\star_h$ therefore
transfers the multiplication of the quantum-plane algebra, or equivalently
operator composition in the represented algebra, to the commutative vector
space $\mathcal P$.

\begin{theorem}[Dilation-induced normal star product]
\label{thm:star-product}
The product $\star_h$ is a unital associative product on $\mathcal P$. For
ordered monomials, it satisfies
\begin{align}
(x^m p^n)\star_h(x^r p^s)
&=
q^{nr}x^{m+r}p^{n+s}.
\label{eq:star-monomial}
\end{align}
More generally, for all $F, G\in\mathcal P$,
\begin{align}
F\star_hG
&=
\sum_{k=0}^{\infty}
\frac{h^k}{k!}
\left(
\mathcal A_p^kF
\right)
\left(
\mathcal A_x^kG
\right).
\label{eq:star-full-series}
\end{align}
For fixed polynomials $F$ and $G$, the right-hand side is entire as a
function of $h$. In particular,
\begin{align*}
p\star_h x
&=
qxp,\\
x\star_h p
&=
xp.
\end{align*}
\end{theorem}

\begin{proof}
Using \eqref{eq:Op-monomial}, we obtain
\begin{align*}
\Op_h(x^mp^n)\Op_h(x^r p^s)=
\mathcal M_{x^m p^n}\mathcal U_x(nh)\mathcal M_{x^rp^s}\mathcal U_x(sh).
\end{align*}
Since
$\mathcal U_x(nh)\mathcal M_{x^r p^s}=
q^{nr}\mathcal M_{x^r p^s}\mathcal U_x(nh)$,
it follows that
\begin{align*}
\Op_h(x^m p^n)
\Op_h(x^r p^s)=q^{nr}\mathcal M_{x^{m+r}p^{n+s}}\mathcal U_x((n+s)h)=
\Op_h\left(q^{nr}x^{m+r}p^{n+s}\right).
\end{align*}
The injectivity of $\Op_h$ proves
\eqref{eq:star-monomial}.

For the same monomials,
$\mathcal A_p^k(x^m p^n)=n^k x^m p^n$,
$\mathcal A_x^k(x^r p^s)=r^k x^r p^s$.
Therefore,
\begin{align*}
\sum_{k=0}^{\infty}
\frac{h^k}{k!}
\left(
\mathcal A_p^k(x^mp^n)
\right)
\left(
\mathcal A_x^k(x^rp^s)
\right)&=
\sum_{k=0}^{\infty}
\frac{(hnr)^k}{k!}
x^{m+r}p^{n+s}\nonumber\\
&=e^{hnr}x^{m+r}p^{n+s}=
q^{nr}x^{m+r}p^{n+s}.
\end{align*}
Bilinearity gives \eqref{eq:star-full-series} for arbitrary polynomials.

The constant polynomial $1_{\mathcal P}$ is the two-sided unit of
$(\mathcal P,\star_h)$, since
$$
1_{\mathcal P}\star_h F
=
F\star_h 1_{\mathcal P}
=
F.
$$
Associativity follows either from the defining product in $\mathfrak A_q$ or
from
\eqref{eq:star-operator} and the associativity of operator composition.
\end{proof}

\begin{remark}
\label{rem:normal-ordering-star-product}
The product $\star_h$ is associated with the ordering convention in which all
powers of $\mathbf x$ are placed to the left of all powers of $\mathbf p$.
A different ordering convention generally produces a different symbol product.
Thus the explicit factor $q^{nr}$ in \eqref{eq:star-monomial} is tied to the
chosen normal-ordering map $\mathcal N_q$. At a broader level, such products
are related to twists and universal deformation formulas
\cite{GiaquintoZhang1998} and to covariant star-product realizations of
quantum spaces \cite{Blohmann2003}. Deformation-theoretic and Hochschild
cohomological properties of the quantum plane and the $q$-Weyl algebra are
studied in \cite{GerstenhaberGiaquinto2014}; those results provide algebraic
background but are not used in the explicit calculation above.
\end{remark}

The exact series immediately separates ordinary multiplication from its
deformation corrections.

\begin{proposition}
\label{prop:star-expansion}
For $F, G\in\mathcal P$,
\begin{align*}
F\star_hG=FG+h\left(\mathcal A_pF\right)\left(\mathcal A_xG\right)+
\frac{h^2}{2}\left(\mathcal A_p^2F\right)
\left(\mathcal A_x^2G\right)+O(h^3)
\end{align*}
as $h\to 0$.
\end{proposition}

\begin{proof}
Retain the terms of orders zero, one, and two in the exact series
\eqref{eq:star-full-series}. Since the series is entire in $h$ for fixed
polynomials, the remaining terms are $O(h^3)$.
\end{proof}

\subsection{Commutative asymptotics of the operator realization}
\label{subsec:operator-expansion}

The exact normal symbol of a represented operator should be distinguished
from the multiplication operator determined by that symbol.  We next quantify
this distinction.  The following expansion compares the dilation
quantization map with ordinary multiplication and is therefore an asymptotic
statement in the deformation parameter, in contrast with the exact symbol
identities above.

\begin{theorem}[First-order commutative expansion of the quantization map]
\label{thm:operator-polynomial-expansion}
Let $F\in\mathcal P$. Fix $h_0>0$ and
$K\Subset\mathbb R_+^2$, and define
\begin{align*}
K_{F,h_0}:=
\left\{
(e^t x,p):
(x,p)\in K,
|t|\le h_0N_p(F)
\right\}.
\end{align*}
Then $K_{F,h_0}\Subset\mathbb R_+^2$.
Moreover, for every $\psi\in\mathcal X$ and every $|h|\le h_0$,
\begin{align*}
&\Op_h(F)\psi=F\psi+h\left(\mathcal A_pF\right)\left(\mathcal A_x\psi\right)+\mathcal R_h^F\psi,\\
&\left\|\mathcal R_h^F\psi\right\|_{C^0(K)}\le\frac{|h|^2}{2}C_{F,K}\left\|\mathcal A_x^2\psi\right\|_{C^0(K_{F,h_0})},\\
&C_{F,K}:=\sum_{m,n}n^2|c_{mn}(F)|\left\|x^mp^n\right\|_{C^0(K)},\,\,
\Op_h(F)=\mathcal M_F+h\mathcal M_{\mathcal A_pF}\mathcal A_x+\mathcal R_h^F.
\end{align*}
After application to any fixed $\psi\in\mathcal X$, the remainder is
$O(h^2)$ locally uniformly on compact subsets of $\mathbb R_+^2$.
\end{theorem}

\begin{proof}
By Proposition~\ref{prop:explicit-quantization}, the action of $\Op_h(F)$ is
given by \eqref{eq:Op-action}.  Fix a pair $(m,n)$ occurring in the polynomial expansion of $F$. For
$(x,p)\in K$, apply Taylor's
formula at $t=0$ to the function
$$
t\longmapsto\psi(e^t x, p).
$$
Its first two derivatives are
\begin{align*}
\frac{d}{dt}\psi(e^t x, p)=\left(\mathcal A_x\psi\right)(e^t x, p),\qquad
\frac{d^2}{dt^2}\psi(e^t x, p)=\left(\mathcal A_x^2\psi\right)(e^t x, p).
\end{align*}
Taking $t=nh$ gives
\begin{align*}
&\psi(q^n x, p)=\psi(x, p)+nh\left(\mathcal A_x\psi\right)(x, p)+R_{n, h}\psi(x, p),\\
&|R_{n, h}\psi(x, p)|\le\frac{n^2|h|^2}{2}\left\|\mathcal A_x^2\psi\right\|_{C^0(K_{F, h_0})}.
\end{align*}
Indeed, $|nh|\le h_0N_p(F)$,
so the entire dilation segment joining $(x, p)$ to $(q^nx, p)$ is contained
in $K_{F, h_0}$.

Substituting this Taylor expansion into \eqref{eq:Op-action}, we obtain
\begin{align*}
\Op_h(F)\psi=
\left(\sum_{m,n}c_{mn}(F)x^mp^n\right)\psi+h\left(\sum_{m, n}n c_{mn}(F)x^mp^n\right)
\mathcal A_x\psi+\mathcal R_h^F\psi.
\end{align*}
Since
$\mathcal A_pF=\sum_{m,n}n c_{mn}(F)x^mp^n$,
the first two terms are $F\psi+h\left(\mathcal A_pF\right)\left(\mathcal A_x\psi\right)$.
Summing the remainder estimates over the finitely many monomials in $F$
gives the stated bound.
\end{proof}

The theorem applies, in particular, to polynomial Hamiltonians.  For
$H\in\mathcal P$, define its quantum-plane and operator realizations by
\begin{align*}
\mathbf H:=\mathcal N_q(H),\,\,
\mathcal H_h:=\Pi(\mathbf H)=\Op_h(H),\,\,
H=\sigma_N(\mathbf H)=\sigma_0(\mathcal H_h).
\end{align*}
Thus $H$, $\mathbf H$, and $\mathcal H_h$ are the symbol, quantum-plane, and
operator representatives of the same Hamiltonian.

\begin{remark}
\label{rem:operator-H-expansion}
Applying Theorem~\ref{thm:operator-polynomial-expansion} with $F=H$ gives
\begin{align*}
\mathcal H_h
&=
\mathcal M_H
+
h\mathcal M_{\mathcal A_pH}\mathcal A_x
+
\mathcal R_h^H,
&
\mathcal R_h^H
&=
O(h^2),\\
\mathcal H_h
&=
\mathcal M_H+O(h),
&
\sigma_0(\mathcal H_h)
&=
H.
\end{align*}
The remainder estimate is understood in the local strong sense specified in
the theorem.  Thus $H$ is the exact normal symbol of $\mathcal H_h$, whereas
$\mathcal M_H$ is only its leading commutative approximation.
\end{remark}

\subsection{Representation--symbol realization of the Hamiltonian action}
\label{subsec:intertwining}

The representation and normal-symbol constructions can now be combined at
the level of Hamiltonian dynamics.  The covariant $q$-derivatives intertwine
with the represented Jackson operators, while ordered multiplication in the
quantum-plane algebra is transported to the normal star product.  These two
facts are precisely the ingredients needed for an exact symbol realization of
the formal Hamiltonian action.

\begin{theorem}[Intertwining of covariant derivatives and Jackson operators]
\label{thm:derivative-intertwining}
Let $\bF\in\mathfrak A_q$ and let $F=\sigma_N(\bF)$.
Then
\begin{align*}
\sigma_N(\qdx\bF)=\Jx F,\qquad
\sigma_N(\qdp\bF)=\Dp F.
\end{align*}
Equivalently,
\begin{align*}
\Pi(\qdx\bF)=\Op_h(\Jx F),\qquad
\Pi(\qdp\bF)=\Op_h(\Dp F).
\end{align*}
\end{theorem}

\begin{proof}
It suffices to verify the identities on the PBW basis. If $m\ge1$, then
Definition~\ref{def:induced-q-derivatives} and the pointwise formula for
$\Jx$ give
\begin{align*}
\sigma_N
\left(
\qdx(\bfx^m\bfp^n)
\right)=[m]_{q^2}x^{m-1}p^n=\Jx(x^mp^n).
\end{align*}
When $m=0$, both sides vanish. Similarly, if $n\ge1$, then
\begin{align*}
\sigma_N
\left(
\qdp(\bfx^m\bfp^n)
\right)=
q^m[n]_{q^2}x^mp^{n-1}=
\Dp(x^mp^n),
\end{align*}
and both sides vanish when $n=0$. Linearity proves the two symbol identities. Applying $\Op_h$ and using
$\Op_h=\Pi\circ\mathcal N_q$ gives the operator identities.
\end{proof}

The theorem gives the commutative diagrams
$$
\begin{array}{ccc}
\mathfrak A_q
&
\overset{\qdx}{\longrightarrow}
&
\mathfrak A_q
\\[2mm]
\downarrow\scriptstyle{\sigma_N}
&&
\downarrow\scriptstyle{\sigma_N}
\\[2mm]
\cP
&
\overset{\Jx}{\longrightarrow}
&
\cP
\end{array}
\qquad
\begin{array}{ccc}
\mathfrak A_q
&
\overset{\qdp}{\longrightarrow}
&
\mathfrak A_q
\\[2mm]
\downarrow\scriptstyle{\sigma_N}
&&
\downarrow\scriptstyle{\sigma_N}
\\[2mm]
\cP
&
\overset{\Dp}{\longrightarrow}
&
\cP.
\end{array}
$$
Thus the Jackson operators are not merely analogous to the formal
$q$-derivatives: they are their exact normal-symbol realizations.

\begin{remark}
\label{rem:operator-derivative-distinction}
The identity
$
\Pi(\qdx\bF)=\Op_h(\Jx F)
$
must not be confused with the operator composition
$\Jx\Pi(\bF)$.
The latter also contains terms in which the derivative operator remains on
the right, as prescribed by the differential relations.  Applying both compositions to the explicitly defined constant function
$\oneX\in\cX$ removes those residual derivative terms:
\begin{align*}
\Jx\Pi(\bF)\oneX=\Jx F,\qquad
\Dp\Pi(\bF)\oneX=\Dp F.
\end{align*}
\end{remark}

Let $\bH, \bF\in\mathfrak A_q$.  We fix the ordering convention inherited
from the covariant calculus and define the formal Hamiltonian action
\begin{align}
\mathbf X_{\bH}[\bF]:=
q^{-1/2}(\qdp\bH)(\qdx\bF)-
q^{1/2}(\qdx\bH)(\qdp\bF).
\label{eq:formal-Hamiltonian-action}
\end{align}
The order of the factors in \eqref{eq:formal-Hamiltonian-action} is part of
the definition.  We use $\mathbf X_{\bH}$ as an ordered Hamiltonian action; we
do not assume that it is an ordinary derivation or a Poisson bracket on the
whole noncommutative algebra.  Its coordinate values reproduce the formal
$q$-Hamiltonian equations, while its behavior on general observables is
analyzed after passing to symbols.

Let
$H=\sigma_N(\bH)$,
$F=\sigma_N(\bF)$.
Define the star-Jackson Hamiltonian action by
\begin{align*}
\mathfrak X_{h,\star}^{H}F:=
q^{-1/2}
(\Dp H)\starh(\Jx F)-
q^{1/2}(\Jx H)\starh(\Dp F).
\end{align*}

The following theorem is the central representation--symbol statement.

\begin{theorem}[Representation--symbol realization of the Hamiltonian action]
\label{thm:exact-Hamiltonian-descent}
For every $\bH,\bF\in\mathfrak A_q$,
\begin{align}
\sigma_N\left(\mathbf X_{\bH}[\bF]\right)=\mathfrak X_{h,\star}^{H}F.
\label{eq:exact-action-descent}
\end{align}
Equivalently,
$\Pi\left(\mathbf X_{\bH}[\bF]\right)=
\Op_h\left(\mathfrak X_{h,\star}^{H}F\right)$.
\end{theorem}

\begin{proof}
By Theorem~\ref{thm:derivative-intertwining},
the symbols of the four derivative factors are
$\Dp H$, $\Jx H$, $\Jx F$, and $\Dp F$.  By the definition of the star
product, the symbol of a product in $\mathfrak A_q$ is the star product of
the two symbols.  Applying this observation to the two ordered products in
\eqref{eq:formal-Hamiltonian-action} proves
\eqref{eq:exact-action-descent}.  The operator identity follows by applying
$\Op_h$.
\end{proof}

Theorem~\ref{thm:exact-Hamiltonian-descent} identifies the symbol-level
Hamiltonian dynamics without a commutative approximation: the ordered
noncommutative product becomes the star product determined by the dilation
representation.  Ordinary pointwise multiplication enters only at the next,
commutativized level.

The coordinate observables are exceptional because their Jackson derivatives
are constants.

\begin{corollary}[Coordinate equations as exact symbol images]
\label{cor:coordinate-descent}
For every Hamiltonian symbol $H\in\cP$,
\begin{align*}
\mathfrak X_{h,\star}^{H}x=q^{-1/2}\Dp H,\qquad
\mathfrak X_{h,\star}^{H}p=-q^{1/2}\Jx H.
\end{align*}
Hence the formal coordinate equations
\begin{align*}
\dot{\bfx}=q^{-1/2}\qdp\bH,\qquad
\dot{\bfp}=-q^{1/2}\qdx\bH
\end{align*}
descend exactly to
\begin{align*}
\dot x=q^{-1/2}\Dp H,\qquad
\dot p=-q^{1/2}\Jx H.
\end{align*}
\end{corollary}

\begin{proof}
The identities
\begin{align*}
&\Jx x=1_{\cP}, &&\Dp x=0,\\
&\Jx p=0,&& \Dp p=1_{\cP}
\end{align*}
give
$\mathfrak X_{h,\star}^{H}x=
q^{-1/2}(\Dp H)\starh 1_{\cP}$,
$\mathfrak X_{h,\star}^{H}p=-q^{1/2}(\Jx H)\starh 1_{\cP}$.
Since $G\starh 1_{\cP}=G$, the conclusion follows.
\end{proof}

The coordinate equations therefore form an exact sector of the symbol
realization.  Their Jackson form is not obtained by taking a small-$h$ limit:
it follows directly from the unit property of the normal star product and the
fact that the derivatives of the coordinate observables are constants.

\begin{remark}
\label{rem:exact-vs-approximate}
The coordinate representation $\Pi$, the normal-ordering isomorphism
$\mathcal N_q$, and the induced star product are exact. Likewise,
Theorem~\ref{thm:exact-Hamiltonian-descent} and
Corollary~\ref{cor:coordinate-descent} are exact statements. Approximation
enters only when the star product is replaced by ordinary multiplication,
when the operator Hamiltonian is replaced by $\mathcal M_H$, or when the
limit $h\to0$ is taken. The coordinate equations are exceptional because one
of the star factors is the polynomial unit $1_{\cP}$, viewed in $\cX$ through $\iota$.
\end{remark}

\section{Pointwise Jackson realizations and geometric defects}
\label{sec:euclidean}

We now turn to the first approximation identified above: replacing the exact
star product by ordinary pointwise multiplication.  This produces Jackson
actions on the commutative function algebra that agree with the exact
star-symbol dynamics on the coordinate observables but differ on nonlinear
ones.  We first compare these actions and their product rules, and then study
ordinary divergence and energy balance for their common coordinate vector
field.  The latter are properties of the commutative realization and should
not be confused with covariance properties of the underlying noncommutative
differential calculus.

\subsection{Pointwise actions and comparison with the star-symbol dynamics}
\label{subsec:euclidean-actions}

There are two natural pointwise actions to compare with the exact
star-symbol action.  The first retains the covariant Jackson derivatives but
uses ordinary multiplication.  The second is an asymmetric pointwise action in which the observable in the
$p$ direction is differentiated by the uncoupled operator $\Jp$ rather than
$\Dp$.  This choice appears in the $q$-Hamiltonian Monte Carlo formulation
\cite{YangDeng2025}, but here it is used only as a comparison of pointwise
calculi.  Both actions reproduce the same coordinate vector field while
defining different actions on nonlinear observables.

For $H,F\in\mathcal X$, define the covariant and asymmetric pointwise
Jackson actions by
\begin{align*}
&\mathfrak X_{h,\mathrm{cov}}^HF:=
q^{-1/2}(\mathcal D_pH)(\mathcal J_xF)
-q^{1/2}(\mathcal J_xH)(\mathcal D_pF),\\
&\mathfrak B_h^HF:=
q^{-1/2}(\mathcal D_pH)(\mathcal J_xF)
-q^{1/2}(\mathcal J_xH)(\mathcal J_pF).
\end{align*}
Set also
\begin{align*}
A_H
&:=q^{-1/2}\mathcal D_pH,
&
B_H
&:=-q^{1/2}\mathcal J_xH,\\
V_h[H]
&:=(A_H,B_H),
&
\mathcal V_h^HF
&:=A_H\partial_xF+B_H\partial_pF.
\end{align*}

\begin{theorem}[Comparison of the exact and pointwise actions]
\label{thm:comparison-actions}
For polynomial $H$ and $F$, the exact star-symbol action, the covariant
pointwise action, the asymmetric pointwise action, and the ordinary
directional action satisfy
\begin{align*}
\mathfrak X_{h,\star}^HF-
\mathfrak X_{h,\mathrm{cov}}^HF
&=
q^{-1/2}
\left[
(\mathcal D_pH)\star_h(\mathcal J_xF)
-(\mathcal D_pH)(\mathcal J_xF)
\right]
\nonumber\\
&-
q^{1/2}
\left[
(\mathcal J_xH)\star_h(\mathcal D_pF)
-(\mathcal J_xH)(\mathcal D_pF)
\right],\\
\mathfrak B_h^HF-
\mathfrak X_{h,\mathrm{cov}}^HF
&=
-q^{1/2}(\mathcal J_xH)(\mathcal J_p-\mathcal D_p)F,\\
\mathfrak X_{h,\mathrm{cov}}^HF-
\mathcal V_h^HF
&=
A_H(\mathcal J_x-\partial_x)F
+B_H(\mathcal D_p-\partial_p)F,\\
\mathfrak B_h^HF-
\mathcal V_h^HF
&=
A_H(\mathcal J_x-\partial_x)F
+B_H(\mathcal J_p-\partial_p)F.
\end{align*}
Locally on compact subsets,
\begin{align*}
\mathfrak X_{h,\star}^HF-
\mathfrak X_{h,\mathrm{cov}}^HF
&=O(h),\\
\mathfrak B_h^HF-
\mathfrak X_{h,\mathrm{cov}}^HF
&=O(h)
\end{align*}
as $h\to0$. All four actions agree exactly on the coordinate observables:
\begin{align*}
\mathfrak X_{h,\star}^Hx
=
\mathfrak X_{h,\mathrm{cov}}^Hx
=
\mathfrak B_h^Hx
=
\mathcal V_h^Hx
&=A_H,\\
\mathfrak X_{h,\star}^Hp
=
\mathfrak X_{h,\mathrm{cov}}^Hp
=
\mathfrak B_h^Hp
=
\mathcal V_h^Hp
&=B_H.
\end{align*}
\end{theorem}

\begin{proof}
The exact identities follow by subtracting the definitions. For fixed polynomial $H$ and $F$, the Jackson derivatives in the first
difference formula remain in finite-dimensional
polynomial spaces and have coefficients locally bounded as $h\to0$.
Therefore Proposition~\ref{prop:star-expansion} gives the stated
star-product estimate. Since
$\mathcal D_p=\mathcal J_p\mathcal U_x(h)$ and
$\mathcal U_x(h)=\Id_{\mathcal X}+h\mathcal A_x+O(h^2)$,
one has
$$
(\mathcal J_p-\mathcal D_p)F
=-h\mathcal J_p\mathcal A_xF+O(h^2),
$$
which gives the second estimate. Finally,
$$
\mathcal J_xx=1,
\qquad
\mathcal D_px=\mathcal J_px=0,
\qquad
\mathcal J_xp=0,
\qquad
\mathcal D_pp=\mathcal J_pp=1,
$$
and star multiplication by $1$ agrees with ordinary multiplication.
\end{proof}

Thus the exact symbol dynamics and all three pointwise realizations produce the
same coordinate vector field. Their differences concern only the calculus of
nonlinear observables. In particular, the asymmetric action already determines the same coordinate field, while
the covariant star action retains the full normal-symbol image of the formal
noncommutative action.

\subsection{Twisted product rules and derivation defects}
\label{subsec:Leibniz-defects}

We next record the product structures of the two pointwise actions. The
one-directional Jackson rules below are $\sigma$-derivation identities, in the
sense that the ordinary Leibniz rule is modified by an algebra endomorphism;
for the general algebraic theory of $\sigma$-derivations and the associated
deformed Lie structures, see \cite{HartwigLarssonSilvestrov2006}. The coupled
operator $\Dp$ involves two commuting shifts and therefore has a corresponding
two-shift product rule. Define the shift operators in
$\mathcal L(\mathcal X)$ by
\begin{align*}
\mathcal S_x:=\Ux(2h),\qquad
\mathcal T_x:=\Ux(h),\qquad
\mathcal S_p:=\Up(2h).
\end{align*}

\begin{proposition}
\label{prop:shifted-product-rules}
For smooth $F,G$,
\begin{align*}
\Jx(FG)
&=
(\Jx F)G
+
(\mathcal S_xF)(\Jx G),\\
\Jp(FG)
&=
(\Jp F)G
+
(\mathcal S_pF)(\Jp G),\\
\Dp(FG)
&=
(\Dp F)(\mathcal T_xG)
+
(\mathcal T_x\mathcal S_pF)(\Dp G).
\end{align*}
\end{proposition}

\begin{proof}
The first two identities follow by adding and subtracting the corresponding
single-shift cross term in the finite-difference numerator.  For the third,
apply the $\Jp$ product rule to
$$
\mathcal T_x(FG)
=
(\mathcal T_xF)(\mathcal T_xG)
$$
and use the commutation of the two dilation groups.
\end{proof}

\begin{theorem}[Derivation defects of the two pointwise actions]
\label{thm:Leibniz-defects}
Define
\begin{align*}
&\cE_{h,\mathrm{cov}}^H(F,G)
:=\Xjack(FG)-(\Xjack F)G-F(\Xjack G),\\
&\cE_{h,\mathrm{asym}}^H(F,G):=
\Bq(FG)-(\Bq F)G-F(\Bq G).
\end{align*}
Then
\begin{align*}
\cE_{h,\mathrm{cov}}^H(F,G)=&A_H
(\mathcal S_xF-F)\Jx G+B_H(\Dp F)(\mathcal T_xG-G)+\\
&B_H(\mathcal T_x\mathcal S_pF-F)\Dp G,
\end{align*}
whereas
\begin{align*}
\cE_{h,\mathrm{asym}}^H(F,G)=A_H(\mathcal S_xF-F)\Jx G+B_H(\mathcal S_pF-F)\Jp G.
\end{align*}
\end{theorem}

\begin{proof}
Insert the corresponding shifted product rules and subtract the ordinary
Leibniz terms.
\end{proof}

The asymmetric pointwise action has a simpler two-shift defect because it uses the
uncoupled derivative $\Jp$ on the observable.  The covariant pointwise action
retains the full differential representation and therefore inherits the
additional preliminary $x$-dilation.

\begin{corollary}
\label{cor:kinetic-potential}
If $H=K(p)$, then $B_H=0$, and both pointwise actions satisfy
\begin{align*}
\Xjack(FG)
&=
(\Xjack F)G
+
(\mathcal S_xF)(\Xjack G),\\
\Bq(FG)
&=
(\Bq F)G
+
(\mathcal S_xF)(\Bq G).
\end{align*}
If $H=U(x)$, then the covariant pointwise action satisfies
\begin{align*}
\Xjack(FG)
&=
(\Xjack F)(\mathcal T_xG)
+
(\mathcal T_x\mathcal S_pF)(\Xjack G),
\end{align*}
whereas the asymmetric pointwise action satisfies
\begin{align*}
\Bq(FG)
&=
(\Bq F)G
+
(\mathcal S_pF)(\Bq G).
\end{align*}
\end{corollary}

\begin{proof}
Only one Jackson component survives in each case, so the conclusions follow
from Proposition~\ref{prop:shifted-product-rules}.
\end{proof}

The preceding results concern the algebraic behavior of the pointwise
actions on observables.  We next turn to the ordinary geometric properties of
the common coordinate vector field.

\subsection{Divergence, energy balance, and model Hamiltonians}
\label{subsec:flow}

Throughout this subsection the Hamiltonian $H$ is assumed real-valued whenever
$V_h[H]$ is interpreted as a flow on $\mathbb R_+^2$.  The preceding algebraic
identities themselves remain valid complex-linearly.

The exact algebraic descent determines the coordinate velocities, but
ordinary divergence and the variation of the classical symbol $H$ along the
resulting flow are properties of the commutative vector field $\Vfield$.
They should not be identified with a formal $q$-symplectic invariance of the
noncommutative differential calculus.

The divergence of $\Vfield$ is
\begin{align*}
\diver\Vfield=
q^{-1/2}\partial_x(\Dp H)-
q^{1/2}\partial_p(\Jx H).
\end{align*}

\begin{proposition}
\label{prop:divergence}
For every sufficiently smooth $H$,
\begin{align}
\diver\Vfield=q^{1/2}\left[\Dp(\partial_xH)-\Jx(\partial_pH)\right].
\label{eq:divergence-exact}
\end{align}
\end{proposition}

\begin{proof}
The operator $\Jp$ commutes with $\partial_x$, while $\partial_x\Ux(h)=q\Ux(h)\partial_x$.
Hence $\partial_x\Dp=q\Dp\partial_x$.
Also, $\partial_p$ commutes with $\Jx$.  Substitution into the preceding expression for $\diver\Vfield$ proves the
formula.
\end{proof}

\begin{corollary}
\label{cor:separable}
If $H(x, p)=U(x)+K(p)$,
then
$\diver\Vfield=0$.
Consequently, wherever the Jackson flow exists, it preserves ordinary
area.
\end{corollary}

\begin{proof}
Both mixed terms in \eqref{eq:divergence-exact} vanish.
\end{proof}

The separable case includes the Hamiltonians commonly used in Hamiltonian
Monte Carlo.  The statement here concerns the continuous Jackson flow; any
volume-preservation property of a discrete time-stepping map is a separate
question.

For monomial Hamiltonians,
\begin{align*}
H_{m,n}(x,p)
&=
x^mp^n,
\qquad
m,n\in\N_0,
\end{align*}
the Jackson derivatives are
\begin{align*}
\Jx H_{m,n}
&=
\begin{cases}
[m]_{q^2}x^{m-1}p^n,
& m\ge1,\\
0,
& m=0,
\end{cases}\\
\Dp H_{m,n}
&=
\begin{cases}
q^m[n]_{q^2}x^mp^{n-1},
& n\ge1,\\
0,
& n=0.
\end{cases}
\end{align*}
Thus no negative powers are introduced when one exponent is zero. The
coordinate equations are
\begin{align}
\dot x
&=
\begin{cases}
q^{m-1/2}[n]_{q^2}x^mp^{n-1},
& n\ge1,\\
0,
& n=0,
\end{cases}
\notag\\
\dot p
&=
\begin{cases}
-q^{1/2}[m]_{q^2}x^{m-1}p^n,
& m\ge1,\\
0,
& m=0.
\end{cases}
\label{eq:monomial-flow}
\end{align}

\begin{proposition}
\label{prop:monomial-divergence}
If $m, n\ge1$, then
\begin{align*}
\diver V_h[H_{m,n}]=q^{1/2}\left(m q^{m-1}[n]_{q^2}-n[m]_{q^2}\right)x^{m-1}p^{n-1}.
\end{align*}
If $m=0$ or $n=0$, the divergence is identically zero. Consequently, every
one-variable monomial and the bilinear Hamiltonian $H_{1,1}=xp$ generate
divergence-free Jackson flows.
\end{proposition}

\begin{proof}
For $m,n\ge1$, differentiate the two coordinate equations in
\eqref{eq:monomial-flow}. If $m=0$ or $n=0$, one component of the field is
independent of the variable with respect to which it is differentiated, and
the divergence vanishes directly.
\end{proof}

The variation of the classical symbol along the Jackson flow is
\begin{align*}
\frac{d}{dt}H(x(t),p(t))=
\nabla H\cdot\Vfield=
q^{-1/2}
(\partial_xH)(\Dp H)-
q^{1/2}
(\partial_pH)(\Jx H).
\end{align*}
It need not vanish for fixed $q\ne1$.

Consider the normalized oscillator
\begin{align*}
H_{\mathrm{osc}}(x, p)=\frac{x^2+p^2}{[2]_{q^2}}.
\end{align*}
Its Jackson coordinate equations are
\begin{align*}
\dot x=
q^{-1/2}p,\qquad
\dot p=-q^{1/2}x.
\end{align*}
The vector field is divergence free because the Hamiltonian is separable.
However, its pointwise Jackson action is not an ordinary derivation.  For
example,
\begin{align*}
\cE_{h,\mathrm{cov}}^{H_{\mathrm{osc}}}(x,p)
&=
-q^{1/2}(q-1)x^2.
\end{align*}
Thus ordinary area preservation and the ordinary Leibniz property are
logically independent.  These finite-$q$ defects lead naturally to the
question of whether the classical Hamiltonian structure is recovered as the
deformation is removed.

\section{Classical limit of the representation and dynamics}
\label{sec:limit}

We therefore examine the limit
$$
q=e^h\longrightarrow1
$$
simultaneously at the operator, symbol, and dynamical levels.  The preceding
construction separates three statements: the operator realization approaches
multiplication by its normal symbol, the normal star product approaches
pointwise multiplication, and the Jackson coordinate field approaches the
classical Hamiltonian vector field.  The estimates below place these limits on
a common local scale.

Let $K\Subset\R_+^2$.  All estimates below are understood on a slightly
larger neighborhood containing the required dilation orbits.

\begin{lemma}
\label{lem:Jackson-expansions}
For a sufficiently smooth function $F$,
\begin{align*}
\Jx F&=\partial_xF+h x\partial_{xx}F+O(h^2),\\\
\Dp F&=\partial_pF+h\left(x\partial_{xp}F+p\partial_{pp}F\right)+O(h^2)
\end{align*}
locally uniformly as $h\to0$.
\end{lemma}

\begin{proof}
Expand $F(e^{2h}x,p)$ and $F(x,e^{2h}p)$ along the corresponding dilation
orbits.  The first formula follows after dividing by
$(e^{2h}-1)x$.  For the second, combine the expansion of $\Jp$ with
$$
\Ux(h)F=F+h\Ax F+O(h^2).
$$
\end{proof}

The classical Hamiltonian vector field and its action are
\begin{align*}
\Vzero&=(\partial_pH,-\partial_xH),\\
\mathcal V_0^HF&=(\partial_pH)\partial_xF-(\partial_xH)\partial_pF.
\end{align*}

\begin{theorem}[Classical limit of the realized dynamics]
\label{thm:three-level-limit}
Let $H$ and $F$ be sufficiently smooth on a neighborhood of
$K\Subset\R_+^2$. Then, as $h\to0$,
\begin{align*}
\Vfield
&=
\Vzero
+
O(h)
\qquad\text{in }C^0(K),\\
\Xjack F
&=
\mathcal V_0^HF
+
O(h)
\qquad\text{in }C^0(K),\\
\Bq F
&=
\mathcal V_0^HF
+
O(h)
\qquad\text{in }C^0(K).
\end{align*}
If $H,F\in\mathcal P$, then
\begin{align*}
\Xstar F
&=
\mathcal V_0^HF
+
O(h)
\qquad\text{in }C^0(K).
\end{align*}
If, in addition, $H\in\mathcal P$, $\mathcal H_h=\Op_h(H)$, and
$\psi\in\mathcal X$, then
\begin{align*}
\mathcal H_h\psi
&=
H\psi
+
O(h)
\qquad\text{in }C^0(K).
\end{align*}
\end{theorem}

\begin{proof}
The field and pointwise action limits follow from
Lemma~\ref{lem:Jackson-expansions} and
$$
q^{\pm1/2}
=
1+O(h).
$$
For polynomial symbols, the star-action limit follows additionally from
Proposition~\ref{prop:star-expansion}. The operator-Hamiltonian limit is the
specialization of
Theorem~\ref{thm:operator-polynomial-expansion} to the polynomial symbol
$F=H$.
\end{proof}

A more explicit first-order expansion of the Jackson field is
\begin{align*}
V_h[H]
&=
V_0[H]+hW_H+O(h^2),\\\
W_H^{(x)}
&=
x\partial_{xp}H
+
p\partial_{pp}H
-
\frac12\partial_pH,\\
W_H^{(p)}
&=
-x\partial_{xx}H
-
\frac12\partial_xH.
\end{align*}

The vector-field estimate also controls the corresponding flows on finite time intervals.

\begin{theorem}[Finite-time convergence of Jackson trajectories]
\label{thm:trajectory-limit}
Let $K\Subset\R_+^2$, and suppose that the classical solution
$$
z_0(t)=(x_0(t),p_0(t))
$$
of
$$
\dot z_0=V_0[H](z_0)
$$
with initial value $z_{\mathrm{in}}$ remains in the interior of $K$ for
$0\le t\le T$.  Assume that $H\in C^3$ is real-valued on a neighborhood of $K$ and that,
for all sufficiently small $|h|$, the solution
$$
\dot z_h=V_h[H](z_h),
\qquad
z_h(0)=z_{\mathrm{in}},
$$
also remains in $K$ on $[0,T]$.  Then there exists $C_T>0$ such that
\begin{align*}
\sup_{0\le t\le T}
|z_h(t)-z_0(t)|
&\le
C_T|h|.
\end{align*}
\end{theorem}

\begin{proof}
On $K$, Theorem~\ref{thm:three-level-limit} gives
$$
\|V_h[H]-V_0[H]\|_{C^0(K)}
\le C|h|.
$$
The classical vector field is Lipschitz on $K$, with some constant $L$.
Therefore,
\begin{align*}
|z_h(t)-z_0(t)|
&\le
C|h|t
+
L\int_0^t
|z_h(s)-z_0(s)|\,ds.
\end{align*}
Gr\"onwall's inequality yields
$$
|z_h(t)-z_0(t)|
\le
C|h|
\frac{e^{Lt}-1}{L},
$$
with the usual interpretation when $L=0$.
\end{proof}

\begin{corollary}
\label{cor:defect-orders}
Let $H,F,G\in C^3$ on a neighborhood of $K\Subset\mathbb R_+^2$.
Then, as $h\to0$,
\begin{align*}
&\|\diver\Vfield\|_{C^0(K)}=O(|h|),
&\|\nabla H\cdot\Vfield\|_{C^0(K)}=O(|h|),\\
&\|\cE_{h,\mathrm{cov}}^H(F,G)\|_{C^0(K)}=O(|h|),
&\|\cE_{h,\mathrm{asym}}^H(F,G)\|_{C^0(K)}=O(|h|).
\end{align*}
\end{corollary}

\begin{proof}
For the divergence, Proposition~\ref{prop:divergence} and
Lemma~\ref{lem:Jackson-expansions}, applied to $\partial_xH$ and
$\partial_pH$, give
\begin{align*}
\Dp(\partial_xH)=\partial_p\partial_xH+O(h),\qquad
\Jx(\partial_pH)=\partial_x\partial_pH+O(h)
\end{align*}
locally on $K$.  Since the mixed derivatives agree and $q^{1/2}=1+O(h)$, formula
\eqref{eq:divergence-exact} gives the stated divergence estimate.  For the
energy, Theorem~\ref{thm:three-level-limit} gives $V_h[H]=V_0[H]+O(h)$, while
$$
\nabla H\cdot V_0[H]=0,
$$
which proves the stated energy estimate.  Finally, both Leibniz estimates follow
from Theorem~\ref{thm:Leibniz-defects}, because every dilation difference
appearing there is $O(h)$ locally on compact sets.
\end{proof}

\begin{remark}[Covariance versus approximation order]
A centered dilation difference, for example, satisfies
\begin{align*}
\Jx^{\mathrm{sym}}=
\frac{1}{e^h-e^{-h}}\Mx^{-1}\left(\Ux(h)-\Ux(-h)
\right),\qquad
\Jx^{\mathrm{sym}}F=\partial_xF+O(h^2).
\end{align*}
Its higher approximation order does not make it interchangeable with the
forward Jackson operator in the present construction.  The forward choice is
dictated by the exact covariant differential relations of
Theorem~\ref{thm:full-differential-representation}; changing the difference
operator changes the represented calculus itself.
\end{remark}

\section{Conclusion}
\label{sec:conclusion}

We have established a representation--symbol correspondence for covariant $q$-Hamiltonian mechanics on the quantum plane. Multiplication and dilation operators give a faithful representation of the coordinate algebra, the same dilation structure realizes the covariant derivatives through Jackson operators, and normal ordering transports operator composition to an associative star product on polynomial symbols. Within this framework the ordered Hamiltonian action has an exact symbol image: it becomes a star-Jackson action, and its restriction to the coordinate observables yields the Jackson coordinate equations without a small-$h$ approximation. Thus the algebraic passage from the covariant quantum-plane calculus to the coordinate dynamics is exact at the symbol level.

Pointwise and classical dynamics arise only after this exact correspondence has been established. Replacing star multiplication by ordinary multiplication changes the action on nonlinear observables, while replacing the operator Hamiltonian by multiplication with its normal symbol is likewise asymptotic. The associated pointwise vector field has ordinary geometric properties---including divergence and energy variation---that need not coincide with formal covariance properties of the noncommutative calculus. These defects vanish to first order as $q\to1$, and both the Jackson vector field and its finite-time trajectories converge to classical Hamiltonian dynamics. The framework therefore makes explicit four connected levels: the quantum-plane algebra, its operator representation, the normal star-symbol dynamics, and the pointwise Jackson dynamics.

Natural extensions include analytic symbol classes, higher-dimensional and multiparameter quantum spaces, alternative covariant calculi and ordering prescriptions, and the relation between the induced pointwise dynamics and genuinely braided or noncommutative Hamiltonian structures.

\section*{Declaration of competing interest}
The authors declare that they have no known competing financial interests or personal relationships that could have appeared to influence the work reported in this paper.

\section*{Author contributions}
Both authors contributed equally to this work and approved the final manuscript.

\bibliographystyle{elsarticle-num-names}
\bibliography{myref}

@article{BauerWachter2003,
  author  = {Bauer, Claudia and Wachter, Hartmut},
  title   = {Operator representations on quantum spaces},
  journal = {Eur. Phys. J. C},
  volume  = {31},
  number  = {2},
  pages   = {261--275},
  year    = {2003},
  doi     = {10.1140/epjc/s2003-01324-0}
}

@article{Blohmann2003,
  author  = {Blohmann, Christian},
  title   = {Covariant realization of quantum spaces as star products by {Drinfeld} twists},
  journal = {J. Math. Phys.},
  volume  = {44},
  number  = {10},
  pages   = {4736--4755},
  year    = {2003},
  doi     = {10.1063/1.1602553}
}

@article{CabanEtAl1994,
  author  = {Caban, P. and Dobrosielski, A. and Krajewska, A. and Walczak, Z.},
  title   = {On {$q$}-deformed {Hamiltonian} mechanics},
  journal = {Phys. Lett. B},
  volume  = {327},
  number  = {3--4},
  pages   = {287--292},
  year    = {1994},
  doi     = {10.1016/0370-2693(94)90730-7}
}

@article{GerstenhaberGiaquinto2014,
  author  = {Gerstenhaber, Murray and Giaquinto, Anthony},
  title   = {On the cohomology of the {Weyl} algebra, the quantum plane, and the {$q$}-{Weyl} algebra},
  journal = {J. Pure Appl. Algebra},
  volume  = {218},
  number  = {5},
  pages   = {879--887},
  year    = {2014},
  doi     = {10.1016/j.jpaa.2013.10.006}
}

@article{GiaquintoZhang1998,
  author  = {Giaquinto, Anthony and Zhang, James J.},
  title   = {Bialgebra actions, twists, and universal deformation formulas},
  journal = {J. Pure Appl. Algebra},
  volume  = {128},
  number  = {2},
  pages   = {133--151},
  year    = {1998},
  doi     = {10.1016/S0022-4049(97)00041-8}
}

@article{HartwigLarssonSilvestrov2006,
  author  = {Hartwig, Jonas T. and Larsson, Daniel and Silvestrov, Sergei D.},
  title   = {Deformations of {Lie} algebras using {$\sigma$}-derivations},
  journal = {J. Algebra},
  volume  = {295},
  number  = {2},
  pages   = {314--361},
  year    = {2006},
  doi     = {10.1016/j.jalgebra.2005.07.036}
}

@article{IpClassical2013,
  author  = {Ip, Ivan Chi-Ho},
  title   = {The classical limit of representation theory of the quantum plane},
  journal = {Int. J. Math.},
  volume  = {24},
  number  = {4},
  pages   = {1350031},
  year    = {2013},
  doi     = {10.1142/S0129167X13500316}
}

@article{IyerMcCune2003,
  author  = {Iyer, Uma N. and McCune, Timothy C.},
  title   = {Quantum differential operators on the quantum plane},
  journal = {J. Algebra},
  volume  = {260},
  number  = {2},
  pages   = {577--591},
  year    = {2003},
  doi     = {10.1016/S0021-8693(03)00052-8}
}

@book{KacCheung2002,
  author    = {Kac, Victor and Cheung, Pokman},
  title     = {Quantum Calculus},
  series    = {Universitext},
  publisher = {Springer},
  address   = {New York},
  year      = {2002},
  doi       = {10.1007/978-1-4613-0071-7}
}

@book{KlimykSchmudgen1997,
  author    = {Klimyk, Anatoli and Schm{\"u}dgen, Konrad},
  title     = {Quantum Groups and Their Representations},
  series    = {Theoretical and Mathematical Physics},
  publisher = {Springer},
  address   = {Berlin},
  year      = {1997},
  doi       = {10.1007/978-3-642-60896-4}
}

@article{Lavagno2006,
  author  = {Lavagno, A. and Scarfone, A. M. and Swamy, P. Narayana},
  title   = {Classical and quantum {$q$}-deformed physical systems},
  journal = {Eur. Phys. J. C},
  volume  = {47},
  number  = {1},
  pages   = {253--261},
  year    = {2006},
  doi     = {10.1140/epjc/s2006-02557-y}
}

@article{LukinSternYakushin1993,
  author  = {Lukin, M. and Stern, A. and Yakushin, I.},
  title   = {Lagrangian and {Hamiltonian} formalism on a quantum plane},
  journal = {J. Phys. A Math. Gen.},
  volume  = {26},
  number  = {19},
  pages   = {5115--5132},
  year    = {1993},
  doi     = {10.1088/0305-4470/26/19/039}
}

@article{LuntsRosenberg1997,
  author  = {Lunts, Valery A. and Rosenberg, Alexander L.},
  title   = {Differential operators on noncommutative rings},
  journal = {Selecta Math. (N.S.)},
  volume  = {3},
  number  = {3},
  pages   = {335--359},
  year    = {1997},
  doi     = {10.1007/s000290050014}
}

@book{Majid1995,
  author    = {Majid, Shahn},
  title     = {Foundations of Quantum Group Theory},
  publisher = {Cambridge University Press},
  address   = {Cambridge},
  year      = {1995},
  doi       = {10.1017/CBO9780511613104}
}

@book{Manin2018,
  author    = {Manin, Yuri I.},
  title     = {Quantum Groups and Noncommutative Geometry},
  edition   = {2nd},
  series    = {CRM Short Courses},
  publisher = {Springer},
  address   = {Cham},
  year      = {2018},
  doi       = {10.1007/978-3-319-97987-8}
}

@misc{MelongWulkenhaar2026,
  author = {Melong, Fridolin and Wulkenhaar, Raimar},
  title  = {{$q$}-deformed topological recursion: quantum curves and non-perturbative analysis},
  year   = {2026},
  note   = {arXiv:2608.02179 [math-ph] (preprint)},
  doi    = {10.48550/arXiv.2608.02179}
}

@article{OstrovskyiSchmudgen2014,
  author  = {Ostrovskyi, Vasyl and Schm{\"u}dgen, Konrad},
  title   = {A resolvent approach to the real quantum plane},
  journal = {Integral Equ. Oper. Theory},
  volume  = {79},
  number  = {4},
  pages   = {451--476},
  year    = {2014},
  doi     = {10.1007/s00020-014-2165-6}
}

@article{Schmudgen2002,
  author  = {Schm{\"u}dgen, Konrad},
  title   = {On the quantum quarter plane and the real quantum plane},
  journal = {Int. J. Math.},
  volume  = {13},
  number  = {3},
  pages   = {279--321},
  year    = {2002},
  doi     = {10.1142/S0129167X02001307}
}

@article{Ubriaco1992,
  author  = {Ubriaco, Marcelo R.},
  title   = {Non-commutative differential calculus and {$q$}-analysis},
  journal = {J. Phys. A Math. Gen.},
  volume  = {25},
  number  = {1},
  pages   = {169--173},
  year    = {1992},
  doi     = {10.1088/0305-4470/25/1/021}
}

@article{WessZumino1991,
  author  = {Wess, Julius and Zumino, Bruno},
  title   = {Covariant differential calculus on the quantum hyperplane},
  journal = {Nucl. Phys. B Proc. Suppl.},
  volume  = {18},
  number  = {2},
  pages   = {302--312},
  year    = {1991},
  doi     = {10.1016/0920-5632(91)90143-3}
}

@misc{YangDeng2025,
  author = {Yang, Xiaomei and Deng, Zhiliang},
  title  = {{Hamiltonian Monte Carlo} from {$q$}-deformed phase--space mechanics},
  year   = {2025},
  note   = {arXiv:2512.13246 [math.NA] (preprint)},
  doi    = {10.48550/arXiv.2512.13246}
}

\end{document}